\documentclass[12pt]{article}

\usepackage{geometry}
\usepackage{amsmath}
\usepackage{amsthm}
\usepackage{import}
\usepackage{bbm}
\usepackage{tikz}
\usepackage[misc]{ifsym}
\usepackage{comment}

\usepackage[onehalfspacing]{setspace}  

\usepackage[title]{appendix}

\usepackage[utopia]{mathdesign}

\usepackage{sectsty}
\sectionfont{\centering \Large \scshape }
\subsectionfont{\centering \large \scshape}

\usepackage[american]{babel}
\usepackage{csquotes}

\usepackage[style=apa, uniquelist=false, backend=biber, doi=false, isbn=false, url=false]{biblatex}
\DeclareLanguageMapping{american}{american-apa}
\defbibenvironment{bibliography}
{\enumerate{}
  {\setlength{\leftmargin}{\bibhang}%
   \setlength{\itemindent}{-\leftmargin}%
   \setlength{\itemsep}{\bibitemsep}%
   \setlength{\parsep}{\bibparsep}}}
{\endenumerate}
{\item}
\usetikzlibrary{shapes.geometric}
\usetikzlibrary{positioning}



\newcounter{parentnumber}
\theoremstyle{plain}

\newtheorem{proposition}{Proposition}

\newtheorem {lemma}{Lemma}
\theoremstyle{definition}

\usepackage{color}

\newcommand{\N}{\mathbb{N}}
\newcommand{\R}{\mathbb{R}}

\usepackage[colorlinks,citecolor=blue]{hyperref}

\title{\vspace{-1em}Efficient Public Good Provision Between and Within Groups}
\author{ 
\normalsize Chowdhury Mohammad Sakib Anwar\thanks{BDT, University of Winchester. \Letter: Sakib.Anwar@winchester.ac.uk} \and 
\normalsize  Jorge Bruno\thanks{BPP University. \Letter: jorgebruno@bpp.com} \and 
\normalsize Renaud Foucart\thanks{Department of Economics, Lancaster University. \Letter: r.foucart@lancaster.ac.uk} \and 
\normalsize Sonali SenGupta\thanks{\textit{Corresponding Author.} Department of Economics, Queen's University Belfast. \Letter: s.sengupta@qub.ac.uk}
}
\date{\normalsize \today}

\begin{document} 
\maketitle
\begin{abstract}
We generalize the model of \textcite{gallice2019co} to incorporate a public goods game with groups, position uncertainty, and observational learning. Contributions are simultaneous within groups, but groups play sequentially based on their observation of an incomplete sample of past contributions. We show that full cooperation between and within groups is possible with self-interested players on a fixed horizon. Position uncertainty implies the existence of an equilibrium where groups of players conditionally cooperate in the hope of influencing further groups. Conditional cooperation implies that each group member is pivotal, so that efficient simultaneous provision within groups is an equilibrium.
    \medskip
\begin{flushleft}\textbf{Keywords} : Public Goods, Groups, Position Uncertainty, Voluntary Contributions.\end{flushleft}\par
\begin{flushleft}\textbf{JEL Codes}: C72, D82, H41 \ \end{flushleft}\par
\end{abstract}

\newpage
\section{Introduction}


When players are organized into groups, such as federations of countries, the provision of public goods faces two different obstacles. Firstly, each group has an incentive to free-ride on the contributions of others. Secondly, players lack incentives to cooperate within groups. 

In this paper, we identify fully cooperative equilibria in a model combining simultaneous contributions to a public good within groups, with a finite-horizon game of sequential decisions and imperfect information among groups. Each group is placed exogenously in a sequence and is unaware of their exact position. Within a group, players individually and simultaneously choose their contribution. All players in a given group observe the total contribution of some of their immediate predecessor groups. 

Our first contribution is to show that position uncertainty - as introduced by \textcite{gallice2019co} with individual players - turns the standard game of private provision of a public good within each group into a variant of \textcite{palfrey1984participation}'s discrete public goods game in which full contribution is part of the equilibrium. Every individual cooperates, recognizing that a single defection may trigger defection by all members of subsequent groups. The mixed strategy equilibrium involves an individual probability of forgiveness increasing with the number of members in a group. This result suggests that large simultaneous group decisions enhance individual cooperation after observing a defection.

Our second contribution is to inform policy debates on the need to build a ``grand coalition'' of countries to contribute to the global public good of decreasing CO2 emissions, and in particular the problem of enforcement within the coalition  \parencite{d1983stability,bloch1996sequential,yi1997stable,belleflamme2000stable,battaglini2016participation,harstad2019compliance,kovavc2021simple}.\footnote{Part of this literature \parencite{baye1996divisionalization,eckert2003negotiating,buchholz2014potentially,foucart2018strategic} has also shown how large groups could delegate decision-making to their individual members in order to increase free-riding on other groups. We find that this result does not hold with conditionally cooperative group members.} In contrast with most of the existing literature, we show that the absence of a grand coalition might be a blessing in disguise as we find that simultaneous individual contributions can happen without any kind of enforcement mechanism as long as there are more than two groups.

Our model belongs more generally to the literature on observational learning, in which individuals sequentially choose an action after seeing predecessors’ choices \parencite{banerjee1992simple,ccelen2004observational,hendricks2012observational,garcia2018consumer,guarino2013social}. Within this literature, it contributes to a subset of papers looking at the effect of information -- in our case on previous contributions -- on contributions to a public good \parencite{figuieres2012vanishing,tajika2020contribute}. Finally, it relates to work on position uncertainty looking at cases where a principal designs the sequential release of information \parencite{doval2020sequential,gershkov2009optimal,nishihara1997resolution}.

\section{The Model} \label{sec:setup}
Let $I=\{1,2,...,N\}$ be a set of players and consider a game with $b\leq N$ many groups composed of players from $I$, where the allocation of players to groups and the sequence of groups are random with equal probability. 


The timing of the game is as follows. First,  Nature (a non-strategic player) assigns players to groups and randomly picks the sequence of groups. Second, groups sequentially play based on partial information on the contribution of their predecessors obtained through a simple sampling rule (which we formally describe below). Within each group, players choose simultaneously and independently whether to contribute. Player $i$ in Group $t$ must choose one of two actions $a_{i,t} \in \{C,D\}$: action $a_{i,t}=C$ implies a contribution of  one unit while $a_{i,t}=D$ implies no contribution. The total group contribution goes towards the common fund. The common fund is then redistributed to all players with a rate of return $r$. We adopt the standard notation $G_{-i}= \sum_{j\neq i} \mathbbm{1}\{a_{j,t}=C\}$ to denote the number of players (other than $i$) who contribute. Payoffs $u_i(a_i,G_{-i})$ of player $i$ is as follows:
\[
    u_i(C,G_{-i}) = \frac{r}{N}(G_{-i}+1)-1
\]
\[
    u_i(D,G_{-i})=\frac{r}{N}(G_{-i}).
\]
We make the standard assumption that the rate of return satisfies $1<r<N$, so that for a fixed $G_{-i}$ the direct gain from contributing is lower than the individual cost, $u_i(C,G_{-i}) \leq u_i(D,G_{-i})$, but contributions are nonetheless socially desirable. The focus of this paper is on the effect of Player $i$'s contribution, or lack thereof, on subsequent players, so that $G_{-i}$ is not fixed.

For $t\leq b$, the symbol $A_t = (a_{i , t})$ denotes actions of the players in Group $t$ and $h_t=(A_t)_{t=1}^{t-1}$ denotes a possible history of actions up to Group $t-1$. Let $H_t$ be the random history at period $t$ with realizations $h_t\in \mathcal{H}_t$ and let $\mathcal{H}_1=\{\emptyset \}$.\footnote{We use period and position interchangeably throughout the paper, as they imply the same in our context.} Players play an extensive form game with imperfect information where each is given a sample $\zeta$ containing the actions of their $m\geq 1$ immediate preceding groups. The value of $m$ is common knowledge. That is, players observe a sample $\zeta=(\zeta',\zeta'')$, where $\zeta'$ states the number of groups sampled and $\zeta''$ states the number of contributors in that sample. A player in Group $t<m$ is provided with a smaller sample $\zeta=(t-1,\zeta'')$, so that players in the first group observe $\zeta_1=(0,0)$ and players in groups positioned between 2 to $m$ observe the actions of all their predecessor groups. They can thus infer their exact position in the sequence from the sample they receive. Formally,  letting $g_{t}=\sum_{Q(i)= t}\mathbbm{1}\{a_{i,t}=C\} $ denote the total contributions in Group $t$, players in Group $t$ receive a sample 
$\zeta_t: \mathcal{H}_t \to \mathcal{S}=\N^2$ containing a tuple:

$$ \zeta_t(h_t)= \Big( \underbrace{\min\{m,t-1\}}_{=\zeta'},\quad \underbrace{\sum_{k=\max\{1,m-t\}}^{t-1}g_{k}}_{=\zeta''}\Big).$$

We use \textcite{kreps1982sequential} sequential equilibrium. Player $i$'s strategy is a function $\sigma_i(C|\zeta): \mathcal{S}\to [0,1]$ that specifies the probability of contributing given the sample received. Let $\sigma=\{\sigma_i\}_{i\in I}$ denote a strategy profile and $\mu=\{\mu_i\}_{i\in I}$ a system of beliefs. A pair $(\sigma, \mu)$ represents an \textit{assessment}. Assessment $(\sigma ^*, \mu^*)$ is a {\it sequential equilibrium} if $\sigma^*$ is sequentially rational given $\mu^*$, and $\mu^*$ is consistent given $\sigma^*$. Let $\mathcal{H}= \cup_{t=1}^{n}\mathcal{H}_t$ be the set of all possible histories. Given a profile of play $\sigma$ let $\mu_i$ denote Player $i$'s beliefs about the history of play : $\mu_i(h|\zeta) : \mathcal{H} \times \mathcal{S}\to [0,1]$, with $\sum_{h \in \mathcal{H}} \mu_i(h|\zeta) =1$ for all $\zeta \in \mathcal{S}$.

\section{Results} \label{sec:results}

The main part of the paper focuses on the case where all groups are of the same size $n = \frac{N}{b}$. We extend our results to the asymmetric case in the Online Appendix. We start by proving the existence of a pure strategy cooperative equilibrium when groups observe strictly more than one of their predecessors. We then focus on the case with a single observation and look at equilibria in pure and mixed strategy. We collect all formal proofs in the Appendix.

Assume a sample size of $m>1$, so that all players observe the contributions of strictly more than one group. Given a sample $\zeta = (\zeta',\zeta'')$ with $m \geq \zeta' > 1$, the simple strategy of ``contributing unless a defection is observed" yields a sequential equilibrium provided that $r$ is large enough. For completeness, in what follows we let $\sigma_i^k$ denote a sequence of strategies with $\sigma_i^k(C\mid \zeta)= 1 - (s_k)$ and $\sigma_i^k(D\mid \zeta) = (s_k)$ where $(s_k)$ is any non-trivial real null sequence, and $\mu_i^k$ the induced belief for strategy $\sigma_i^k$ for each $k\in \N$.
\begin{proposition}
[Pure Strategies with $m>1$]
\label{lem:m>1Symmetric}  Consider the profile of play

\begin{align*}
    \sigma_i^*(C\mid \zeta) =
\begin{cases}
1, & \text{if $\zeta$ contains no defections} \\
0, & \text{otherwise.}
\end{cases}
\end{align*}
It follows that $(\sigma^*,\mu^*)$ is a sequential equilibrium provided that 
\begin{align} r \geq \frac{2N}{N-n(m+1)+2}. \end{align} \label{equ:pure}
\end{proposition}

The existence condition for this cooperative equilibrium in pure strategy is that the rate of return on investment in public goods is sufficiently large for the expected benefits from sustained cooperation by subsequent groups to be higher than the individual cost of contributing. Whenever groups are of size one, the condition in \eqref{equ:pure} is identical to the result in Proposition 1 of \textcite{gallice2019co}. While the introduction of simultaneous choices by group members does not impede the existence of a fully cooperative equilibrium in pure strategy, it does however make the existence condition more stringent, as the right-hand side of \eqref{equ:pure} is increasing in $n$. This reflects the fact that, on top of potentially free-riding on subsequent groups, a player's unilateral deviation from the equilibrium strategy allows her to free-ride on her other team members with certainty.

We now assume that groups only observe the contribution of their immediate predecessor, $m=1<b$. This case should intuitively be the most favorable for the existence of a cooperative equilibrium: the right-hand side of \eqref{equ:pure} is increasing in $m$, making the existence condition more stringent when the size of the sample increases. When $n=1$, however, \textcite{gallice2019co} prove that a pure strategy equilibrium exists for all $r\in [2,3-\frac{3}{N+1}]$, but not for higher values. The reason is that, when $m=1$ and groups are of size one, a player has the power to restore a full history of cooperation after observing a defection, simply by contributing. If $r$ is large enough, the pure strategy equilibrium therefore stops to exist anymore, as there is no credible punishment for defectors.

In contrast, we show that when groups are larger, $n>1$, a pure strategy equilibrium exists for all values of $r\geq \frac{2N}{N-2(n-1)}$. The main difference with $n=1$ is that there is no way to unilaterally restore an entire history of contribution and make subsequent players contribute, even when it would yield higher expected surplus than defecting. Just like in the case with $m>1$, the punishment of defectors in a pure strategy equilibrium is thus credible. We then show that - concurrent to the pure strategy equilibrium - for large enough values of $r$ there exist at least two mixed strategy equilibria where players forgive defections with probability $\gamma\in (0,1)$. 

\begin{proposition}  
[Pure and Mixed Strategies with $m=1<n$]\label{thm:pureStrat} 
For any value of $r\geq \frac{2N}{N-2(n-1)}$ and given the profile of play
\begin{align*}
    \sigma_i^*(C\mid \zeta) =
\begin{cases}
1, & \zeta \in\{(0,0), (1,n)\}\\
0, & \text{otherwise}
\end{cases}
\end{align*}
for all $i\in I$, the assessment $(\sigma^*, \mu^*)$ is a sequential equilibrium.

Moreover, there exists a $r^\sharp<N$ so that for all $r>r^\sharp$ there exist two distinct values $\gamma^1_r,\gamma^2_r \in(0,1)$ where, for all $i\in I$, the profiles of play

\begin{align*}
    \sigma_{i,1}^*(C\mid \zeta) =
\begin{cases}
1, & \zeta \in\{(0,0), (1,n)\}\\
\gamma^1_r, & \text{otherwise},
\end{cases}
\end{align*}

\begin{align*}
    \sigma_{i,2}^*(C\mid \zeta) =
\begin{cases}
1, & \zeta \in\{(0,0), (1,n)\}\\
\gamma^2_r, & \text{otherwise}
\end{cases}
\end{align*}
 establish two distinct sequential equilibria $(\sigma^*_1, \mu^*_1)$ and $(\sigma^*_1, \mu^*_2)$, respectively.
 \end{proposition}
 
The pure strategy equilibrium follows the same logic as Proposition~\ref{lem:m>1Symmetric}. For a mixed strategy that forgives defections with strictly positive probability to exist out-of-equilibrium, all players must be indifferent between contributing or not, when observing a sample with at least one defection. They must therefore balance the future contributions they expect if they contribute, with what they can expect by defecting, given that everyone else forgives defections with probability $\gamma$. 

We illustrate this result in Figure~\ref{fig:1}. The function $\Delta(\gamma)$ describes the difference in utility between contributing and defecting, upon witnessing a defection. A first observation is that whenever groups are at least of size two (solid line in Figure~\ref{fig:1} for $n=5$), $\Delta(0)=\Delta(1)=\frac{r}{N} -1<0$. If a player expects that no one ever forgives defections ($\gamma=0$), there is no way for her to restore a full history of cooperation after observing a defection, as other members of the group will never contribute. If she expects everyone to forgive defections all the time ($\gamma=1$), then there is no benefit from ever cooperating, as it will have no influence on other players' behavior. It follows, by Rolle's Theorem, that $\Delta(\gamma)$ must attain at least one local maximum between $(0,1)$: for smaller values of $\gamma$, a higher probability of forgiveness makes it more profitable to also forgive and restore cooperation, but at some point this effect is compensated by the gains from defecting in the hope the next group will nonetheless contribute. In the proof of Proposition~\ref{thm:pureStrat}, we show that for at least one $r^\sharp<N$ - and thus all other $r>r^\sharp$ -  $\Delta(\gamma)$'s local maximum is positive. In which case, $\Delta(\gamma)$ exhibits at least two roots and, thus, two distinct sequential equilibria. 

\begin{figure}[htbp]
    \centering
    \includegraphics[scale=0.5]{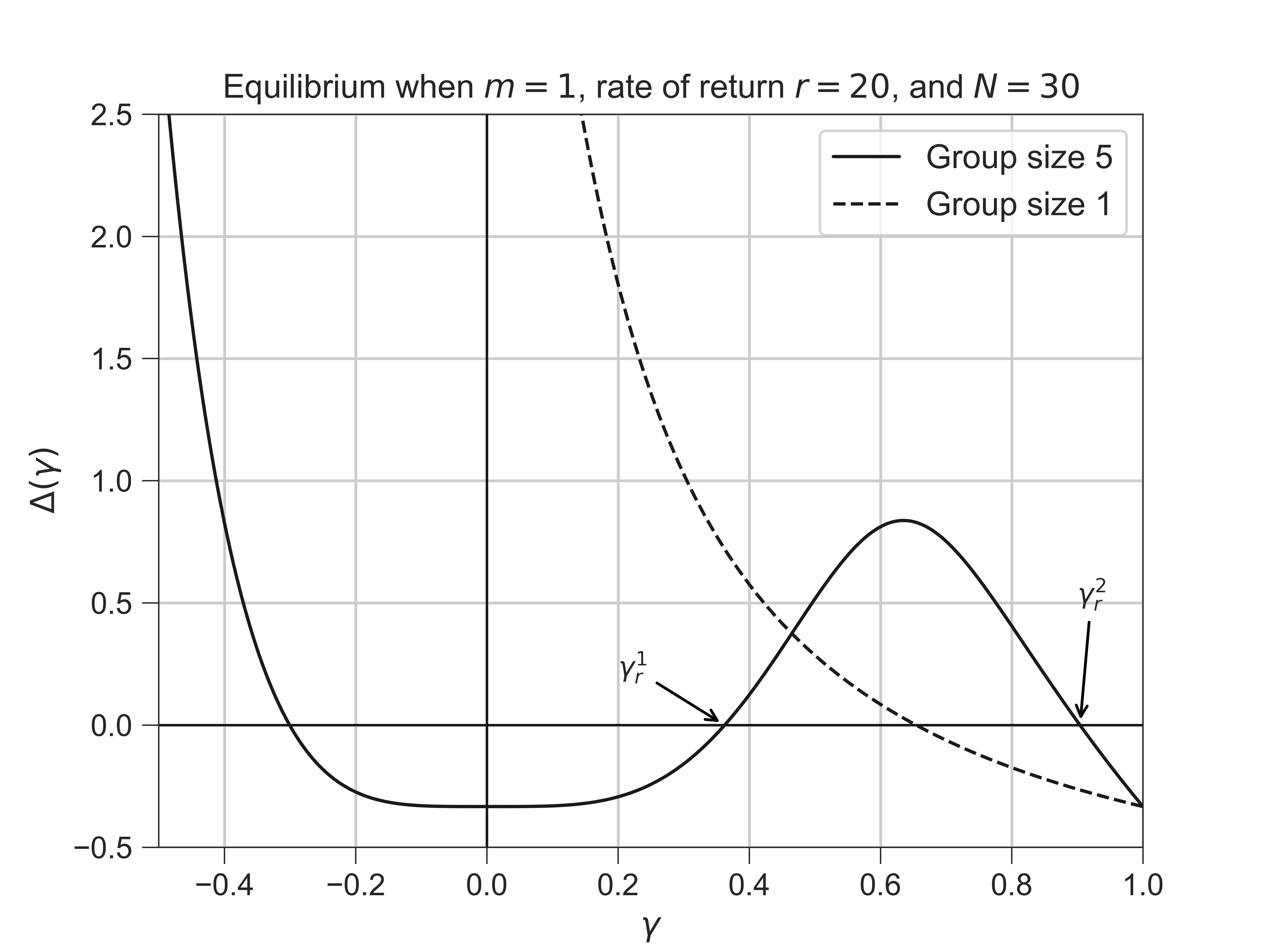}
    \caption{Mixed Strategy Equilibrium}
    \label{fig:1}
\end{figure}

The two intersections between the solid line and the horizontal axis $\Delta(\gamma)=0$ correspond to mixed strategy equilibria of different nature. Assuming that players defect by mistake with a very small probability $\epsilon>0$, there is a self-enforcing coalition \parencite{bernheim1987coalition} of all players who would be better off forgiving with probability $\gamma^2_r$ than with $\gamma^1_r$ or $0$. As group size becomes larger, the two equilibrium points move to the right. This result directly follows from the equilibrium conditions: for a given $\gamma$, the probability for an entire group to contribute after a defection is decreasing with the number of group members. Hence, to make players indifferent between contributing or not, the probability of an individual player forgiving must be higher in larger groups.

 The dashed line represents the case with groups of size one studied by \textcite{gallice2019co}. For the lowest values of $\gamma$, $\Delta(\gamma)$ is high: it pays to forgive a unilateral defection if you expect no one else to do it, because you know it is the only way to restore future cooperation. Then, as $\gamma$ increases, the gain from forgiving decreases, up to the point where all players are indifferent between contributing or not, $\Delta(\gamma)=0$.

\section{Conclusion} \label{sec:conclu}

This paper shows that full contribution to a public good within groups of self-interested players is achievable in a finite horizon. The main mechanism behind the paper is that groups of players believe that providing evidence of their full contribution may help foster future contributions by other groups. In that context, every player within a group is pivotal, so that a standard game of simultaneous contributions becomes a threshold public goods game.

We believe that crowd-funding operations to finance public goods could be inspired by our results. A fundraiser could decide to split their pool of potential donors into several groups. They could then contact groups sequentially and inform each participant of a group-based crowdfunding objective. Subsequent groups (if there are any) would then be informed of whether their immediate predecessors achieved their goal, perhaps with information on the fact that a group's contribution may induce further groups to also contribute. If there is some position uncertainty so that no group is aware of being the last one, our model predict that large contributions are feasible. The same logic could be applied in the context of sequential costly voting.

\section*{Conflict of Interest}
None. 
\section*{Declaration of Generative AI and AI-assisted technologies in the writing process} 
None. 
\section*{Acknowledgments}
We wish to thank Trivikram Dokka, Konstantinos Georgalos, Alexander Matros and Jaideep Roy, as well as participants of the SING17 conference and two anonymous referees for their helpful comments.
\newpage
\printbibliography

@article{kreps1982sequential,
  title={Sequential equilibria},
  author={Kreps, David M and Wilson, Robert},
  journal={Econometrica},
  pages={863--894},
  year={1982},
  publisher={JSTOR}
}

@article{gallice2019co,
  title={Co-operation in social dilemmas through position uncertainty},
  author={Gallice, Andrea and Monz{\'o}n, Ignacio},
  journal={Economic Journal},
  volume={129},
  number={621},
  pages={2137--2154},
  year={2019},
  publisher={Oxford University Press}
}

@article{nishihara1997resolution,
  title={A resolution of N-person prisoners' dilemma},
  author={Nishihara, Ko},
  journal={Economic Theory},
  volume={10},
  number={3},
  pages={531--540},
  year={1997},
  publisher={Springer}
}

@article{tajika2020contribute,
  title={Contribute once! Full efficiency in a dynamic contribution game},
  author={Tajika, Tomoya},
  journal={Games and Economic Behavior},
  volume={123},
  pages={228--239},
  year={2020},
  publisher={Elsevier}
}

@article{gershkov2009optimal,
  title={Optimal voting schemes with costly information acquisition},
  author={Gershkov, Alex and Szentes, Bal{\'a}zs},
  journal={Journal of Economic Theory},
  volume={144},
  number={1},
  pages={36--68},
  year={2009},
  publisher={Elsevier}
}

@article{doval2020sequential,
  title={Sequential information design},
  author={Doval, Laura and Ely, Jeffrey C},
  journal={Econometrica},
  volume={88},
  number={6},
  pages={2575--2608},
  year={2020},
  publisher={Wiley Online Library}
}

@article{hendricks2012observational,
  title={Observational learning and demand for search goods},
  author={Hendricks, Kenneth and Sorensen, Alan and Wiseman, Thomas},
  journal={American Economic Journal: Microeconomics},
  volume={4},
  number={1},
  pages={1--31},
  year={2012}
}

@article{guarino2013social,
  title={Social learning with coarse inference},
  author={Guarino, Antonio and Jehiel, Philippe},
  journal={American Economic Journal: Microeconomics},
  volume={5},
  number={1},
  pages={147--74},
  year={2013}
}

@article{garcia2018consumer,
  title={Consumer search with observational learning},
  author={Garcia, Daniel and Shelegia, Sandro},
  journal={RAND Journal of Economics},
  volume={49},
  number={1},
  pages={224--253},
  year={2018},
  publisher={Wiley Online Library}
}

@article{banerjee1992simple,
  title={A simple model of herd behavior},
  author={Banerjee, Abhijit V},
  journal={Quarterly Journal of Economics},
  volume={107},
  number={3},
  pages={797--817},
  year={1992},
  publisher={MIT Press}
}

@article{figuieres2012vanishing,
  title={Vanishing leadership and declining reciprocity in a sequential contribution experiment},
  author={Figuieres, Charles and Masclet, David and Willinger, Marc},
  journal={Economic Inquiry},
  volume={50},
  number={3},
  pages={567--584},
  year={2012},
  publisher={Wiley Online Library}
}

@article{ccelen2004observational,
  title={Observational learning under imperfect information},
  author={{\c{C}}elen, Bo{\u{g}}a{\c{c}}han and Kariv, Shachar},
  journal={Games and Economic Behavior},
  volume={47},
  number={1},
  pages={72--86},
  year={2004},
  publisher={Elsevier}
}

@article{battaglini2016participation,
  title={Participation and duration of environmental agreements},
  author={Battaglini, Marco and Harstad, B{\aa}rd},
  journal={Journal of Political Economy},
  volume={124},
  number={1},
  pages={160--204},
  year={2016},
  publisher={University of Chicago Press Chicago, IL}
}

@article{foucart2018strategic,
  title={Strategic decentralization and the provision of global public goods},
  author={Foucart, Renaud and Wan, Cheng},
  journal={Journal of Environmental Economics and Management},
  volume={92},
  pages={537--558},
  year={2018},
  publisher={Elsevier}
}

@article{eckert2003negotiating,
  title={Negotiating environmental agreements: Regional or federal authority?},
  author={Eckert, Heather},
  journal={Journal of Environmental Economics and Management},
  volume={46},
  number={1},
  pages={1--24},
  year={2003},
  publisher={Elsevier}
}

@article{kovavc2021simple,
  title={A simple dynamic climate cooperation model},
  author={Kov{\'a}{\v{c}}, Eugen and Schmidt, Robert C},
  journal={Journal of Public Economics},
  volume={194},
  pages={104329},
  year={2021},
  publisher={Elsevier}
}

@article{d1983stability,
  title={On the stability of collusive price leadership},
  author={d'Aspremont, Claude and Jacquemin, Alexis and Gabszewicz, Jean Jaskold and Weymark, John A},
  journal={Canadian Journal of Economics},
  pages={17--25},
  year={1983},
  publisher={JSTOR}
}

@article{buchholz2014potentially,
  title={Potentially harmful international cooperation on global public good provision},
  author={Buchholz, Wolfgang and Cornes, Richard and R{\"u}bbelke, Dirk},
  journal={Economica},
  volume={81},
  number={322},
  pages={205--223},
  year={2014},
  publisher={Wiley Online Library}
}

@article{bloch1996sequential,
  title={Sequential formation of coalitions with fixed payoff division and externalities},
  author={Bloch, Francis},
  journal={Games and Economic Behavior},
  volume={14},
  number={0043},
  pages={90--123},
  year={1996}
}

@article{yi1997stable,
  title={Stable coalition structures with externalities},
  author={Yi, Sang-Seung},
  journal={Games and Economic Behavior},
  volume={20},
  number={2},
  pages={201--237},
  year={1997},
  publisher={Elsevier}
}

@article{belleflamme2000stable,
  title={Stable coalition structures with open membership and asymmetric firms},
  author={Belleflamme, Paul},
  journal={Games and Economic Behavior},
  volume={30},
  number={1},
  pages={1--21},
  year={2000},
  publisher={Elsevier}
}

@article{baye1996divisionalization,
  title={Divisionalization, franchising, and divestiture incentives in oligopoly},
  author={Baye, Michael R and Crocker, Keith J and Ju, Jiandong},
  journal={The American Economic Review},
  pages={223--236},
  year={1996},
  publisher={JSTOR}
}

@article{bernheim1987coalition,
  title={Coalition-proof nash equilibria I. concepts},
  author={Bernheim, B Douglas and Peleg, Bezalel and Whinston, Michael D},
  journal={Journal of Economic Theory},
  volume={42},
  number={1},
  pages={1--12},
  year={1987},
  publisher={Elsevier}
}

@article{harstad2019compliance,
  title={Compliance technology and self-enforcing agreements},
  author={Harstad, B{\aa}rd and Lancia, Francesco and Russo, Alessia},
  journal={Journal of the European Economic Association},
  volume={17},
  number={1},
  pages={1--29},
  year={2019},
  publisher={Oxford University Press}
}

@article{palfrey1984participation,
  title={Participation and the provision of discrete public goods: a strategic analysis},
  author={Palfrey, Thomas R and Rosenthal, Howard},
  journal={Journal of public Economics},
  volume={24},
  number={2},
  pages={171--193},
  year={1984},
  publisher={Elsevier}
}

\newpage
\appendix
\numberwithin{equation}{section} 
\renewcommand{\theequation}{\Alph{section}.\arabic{equation}} 

\renewcommand{\thesection}{Appendix \Alph{section}:} 
\renewcommand{\thesubsection}{\Alph{section}.\arabic{subsection}:} 

\section{Proofs}
\subsection{Proof of Proposition~\ref{lem:m>1Symmetric}}

\begin{proof}
In what follows we let $\sigma_i^k$ denote a sequence of strategies with $\sigma_i^k(C\mid \zeta)= 1 - (s_k)$ and $\sigma_i^k(D\mid \zeta) = (s_k)$, where $(s_k)$ is any non-trivial real null sequence, and put $\mu_i^k$ as the induced belief for strategy $\sigma_i^k$ for each $k\in \N$.

Consider a player in Group $t$ and a history $\zeta = (m,c)$, with $c$ as the number of observed contributions in $m$ sampled groups. Assume the sample contains at least one defection so that $c<mn$. For groups of size at least $2$, a player that observes $\zeta$ is aware that every other player in her group also witnesses a defection and given the pure profile of play, the effect of the defection would extend beyond her group regardless of her contribution, or lack thereof. Even when $n=1$ players in subsequent groups will still witness defections as $m>1$ and $c+1 + kn = mn$ for some $k\in \N$. This means that if the defection occurs in Group $t'<t$ and all other players in groups succeeding $t'$ must defect, then a defection must inevitably appear in Group $t-1$. It follows that defection after defection is optimal given any value of $r$.

In contrast, if $\zeta = (\zeta',\zeta'n)$, meaning all players in the observed sample have contributed, a player in Group $t>m$ prefers to contribute when 
\begin{align} \label{equ:pref}
\frac{r}{N}N -1 \geq \frac{r}{N}\frac{N+(m+1)n-2}{2},
\end{align}
since she expects her position to be in the mid-point between $m+1$ and $b$, and expects all other members of her group to contribute. This inequality simplifies to the condition in \eqref{equ:pure}. Any other player in a Group $t<m$ knows for sure they are at the beginning of the sequence. They therefore have an even greater incentive to contribute as their gains from contributing and encouraging future contributions are larger than those in a Group $t'>m$. 
\end{proof}

\subsection{Proof of Proposition~\ref{thm:pureStrat}}

\begin{proof} Let's assume that we have $b = \frac{N}{n}$ many groups with $n$ individuals per group. Assume that the standard strategy given profile $\zeta$ is
\begin{align*}
    \sigma^*(C\mid \zeta) =
\begin{cases}
1, & \zeta \in \{(0,0), (1,n)\}\\
\gamma, & \text{otherwise},
\end{cases}
\end{align*}
where $\gamma \in[0,1)$ is the - off the equilibrium path - probability with which a player contributes after observing at least one defection in the observed sample. Fix a profile $\overline{\zeta}$ and set for Player $j$ the strategies
\[
\sigma_j^C(\zeta)
=
\begin{cases}
\sigma^*_j(\zeta), & \zeta \not = \overline{\zeta}\\
1, & \zeta = \overline{\zeta}.
\end{cases}
\]

\[
\sigma_j^D(\zeta)
=
\begin{cases}
\sigma^*_j(\zeta), & \zeta \not = \overline{\zeta}\\
0, & \zeta = \overline{\zeta},
\end{cases}
\]
and $\mu_j^D$ and $\mu_j^C$ as their corresponding beliefs. Set $\phi_t(\gamma)$ and $\psi_t(\gamma)$ as the number of additional contributions said player expects from contributing rather than defecting whilst in Group $t$, and the likelihood  of being in Group $t$ after observing a defection, respectively. For $n>1$, $\phi_t(\gamma)$ is different depending on the history $\overline{\zeta} = (1, n')$ observed by Player $j$. Indeed, if $\overline{\zeta} = (1, n')$ with $n>n'$ then all other member of Group $t$' will also observe a defection and act according to their strategies. Hence, there is a non-zero probability that subsequent groups will also witness a sample with defection even if Player $j$ itself contributes. In contrast, if $\overline{\zeta} = (1, n)$ then Player $j$'s defection will be the only one in her group and the effect of her defection should be larger than the previous scenario. The following lemma demonstrates just that.

\begin{lemma}\label{lem:psiPhi} Both $\phi_t(\gamma), \psi_t(\gamma): [0,1] \to \R$ are continuous functions where:
\begin{enumerate}
\item given $\overline{\zeta} = (1, n')$ with $n>n'$
\[ \phi_t(\gamma) = \frac{n(1-\gamma)(1-(1-\gamma^n)^{b-t})}{\gamma}+1,
\]
with $\phi_t(0) = 1$ for all $n>1$ and $\phi_t(0) = b-t+1$ for $n=1$; 
\item given $\overline{\zeta} = (1, n)$

\[ \phi_t(\gamma) = \frac{n(1-\gamma)(1-(1-\gamma^n)^{b-t})}{\gamma^n}+1,
\]
with $\phi_t(0) = (b-t)n +1$; and
\item  \[
\psi_t(\gamma) = \frac{(1-(1-\gamma^n)^{t-1})}{b-1-\gamma^{-n}(1-\gamma^n)(1-(1-\gamma^n)^{b-1})}
\]
with $\psi_t(0)  =\frac{2(t-1)}{b(b-1)}$.
\end{enumerate}
\end{lemma}

\begin{proof}

\noindent
The number of additional contributors a player in Group $t$ should expect from contributing rather than defecting given history $\overline{\zeta}$ becomes 
\[
\phi_t(\gamma) = E_{\mu^C}(G_{-j}\mid \zeta = \overline{\zeta}, Q(j) = t) - E_{\mu^D}(G_{-j}\mid \zeta = \overline{\zeta}, Q(j) = t),
\]
where
\begin{align*}
E_{\mu^C}(G_{-j}\mid \zeta = \overline{\zeta}, Q(j) = t)  &= \sum_{i=1}^{t-1}E_{\mu^*}(G_{i}\mid \zeta = \overline{\zeta}, Q(j) = t)\\
& +  \sum_{i=t}^{b} E_{\mu^C}(G_{i}\mid \zeta = \overline{\zeta}, Q(j) = t),
\end{align*}
\begin{align*}
E_{\mu^D}(G_{-j}\mid \zeta = \overline{\zeta}, Q(j) = t)  &= \sum_{i=1}^{t-1}E_{\mu^*}(G_{i}\mid \zeta = \overline{\zeta}, Q(j) = t)\\
& +  \sum_{i=t}^{b} E_{\mu^D}(G_{i}\mid \zeta = \overline{\zeta}, Q(j) = t),
\end{align*}
and $G_i$ represents the $i^\text{th}$ group.
Hence,
\[
\phi_t(\gamma) = \sum_{i=t}^{b} E_{\mu^C}(G_{i}\mid \zeta = \overline{\zeta}, Q(j) = t) -   E_{\mu^D}(G_{i}\mid \zeta = \overline{\zeta}, Q(j) = t).
\]
\noindent
1. If the sample $\overline{\zeta} = (1,n')$ contains a defection (i.e., $n>n'$) we have
\[
E_{\mu^D}(G_{t}\mid \zeta = \overline{\zeta}, Q(j) = t) =   (n-1)\gamma 
\]
and
\[
E_{\mu^C}(G_{i}\mid \zeta = \overline{\zeta}, Q(j) = t) =  E_{\mu^D}(G_{i}\mid \zeta = \overline{\zeta}, Q(j) = t) + 1.
\]
\noindent
For $t+1$ we get 
\[
E_{\mu^D}(G_{t+1}\mid \zeta = \overline{\zeta}, Q(j) = t)= \gamma n 
\]
\text{ and } 
\begin{align*}
E_{\mu^C}(G_{t+1}\mid \zeta = \overline{\zeta}, Q(j) = t) &= \gamma^{n-1}n + (1-\gamma^{n-1})E_{\mu^D}(G_{t+1}\mid \zeta = \overline{\zeta}, Q(j) = t)\\
&= \gamma^{n-1}n + (1-\gamma^{n-1})\gamma n.
\end{align*}
In general, for $t+k$ with $k\geq 1$ we have
\[
E_{\mu^D}(G_{t+k}\mid \zeta = \overline{\zeta}, Q(j) = t) = n\left[1-(1-\gamma)(1-\gamma^n)^{k-1}\right] 
\]
and 

\[
E_{\mu^C}(G_{t+k}\mid \zeta = \overline{\zeta}, Q(j) = t) = n\gamma^{n-1}+ (1-\gamma^{n-1})E_{\mu^D}(G_{t+k}\mid \zeta = \overline{\zeta}, Q(j) = t).
\]
In turn, 
\[
E_{\mu^C}(G_{t+k}\mid \zeta = \overline{\zeta}, Q(j) = t) - E_{\mu^D}(G_{t+k}\mid \zeta = \overline{\zeta}, Q(j) = t) = n\gamma^{n-1}(1-\gamma)(1-\gamma^n)^{k-1}.
\]
As a sum of powers of $(1-\gamma^n)^{k-1}$, we deduce that for any $\gamma \in (0,1]$:
\[
\phi_t(\gamma) = \frac{n(1-\gamma)(1-(1-\gamma^n)^{b-t})}{\gamma}+1,
\]
with
\[
\lim_{\gamma\to 0} \phi_t(\gamma) = \lim_{\gamma\to 1} \phi_t(\gamma) = 1,
\]
for all $n>1$. We observe that $\phi_t(\gamma)>1$ for all $
\gamma \in (0,1)$. In fact, the distribution $\phi_t(\gamma)$ is bell shaped. This means that there is an optimal value of $\gamma$ that maximizes the additional contributions a player can expect by contributing rather than defecting. If we let $n=1$ we get 
\[
\phi_t(\gamma) =\frac{1-(1-\gamma)^{b-t+1}}{\gamma}
\]
with 
\[
\lim_{\gamma\to 0} \phi_t(\gamma) = b-t+1.
\]
This is precisely what is obtained in \textcite{gallice2019co}.\\

2. If the sample $\overline{\zeta} = (1,n)$, then the computation is similar but much simpler. For $t+k$ with $k\geq 1$ we have
\[
E_{\mu^D}(G_{t+k}\mid \zeta = \overline{\zeta}, Q(j) = t) = n\left[1-(1-\gamma^n)^{k-1}(1-\gamma)\right] 
\]
and 

\[
E_{\mu^C}(G_{t+k}\mid \zeta = \overline{\zeta}, Q(j) = t) = n. 
\]
In turn, 
\[
E_{\mu^C}(G_{t+k}\mid \zeta = \overline{\zeta}, Q(j) = t) - E_{\mu^D}(G_{t+k}\mid \zeta = \overline{\zeta}, Q(j) = t) = n(1-\gamma)(1-\gamma^n)^{k-1}.
\]
Therefore, for any $\gamma \in (0,1]$ we have
\[
\phi_t(\gamma) = \frac{n(1-\gamma)(1-(1-\gamma^n)^{b-t})}{\gamma^n}+1,
\]
with
\[
\lim_{\gamma\to 0} \phi_t(\gamma) = (b-t)n+1.
\]

\noindent

3. Let $\psi_t(\gamma)$ denote the likelihood that a player finds itself in position $t$ after witnessing a defection. It follows that
\begin{align*}
\psi_t(\gamma) = \frac{\sum_{j=0}^{t-2}(1-\gamma^n)^j}{\sum_{k=2}^b\sum_{i=0}^{k-2}(1-\gamma^n)^{i}} &= \frac{ \gamma^{-n}(1-(1-\gamma^n)^{t-1})}{\gamma^{-n}\left(b-1+\gamma^{-n}(1-\gamma^n)(1-(1-\gamma^n)^{b-1}\right)}\\
&= \frac{(1-(1-\gamma^n)^{t-1})}{b-1-\gamma^{-n}(1-\gamma^n)(1-(1-\gamma^n)^{b-1})}.
\end{align*}

For all other $n>1$ we can make the replacement $y = \gamma^p$ to obtain, after applying L'Hospital's Rule, $\psi_t(0) =\frac{2(t-1)}{b(b-1)}$.\\
\end{proof}

\noindent

In what follows it becomes useful to let $\phi_t(\gamma, \zeta)$ denote $\phi_t(\gamma)$ given a profile $\zeta$. A player in Group $1$ contributes whenever $\frac{r}{N}\phi_1(\gamma, (0,0)) - 1 \geq 0$; a player who received a sample $\overline{\zeta} = (1,n)$ contributes provided $\sum_{t=2}^b\frac{1}{b-1}\phi_t(\gamma, \overline{\zeta})-1\geq 0$; and given profile $\overline{\zeta}' = (1,n')$ with $n'<n$ a player contributes provided $\sum_{t=2}^b \psi_t(\gamma)\phi_t(\gamma, \overline{\zeta}')-1\geq 0$.

\begin{lemma}\label{lem:inequalities} Given profiles $\overline{\zeta}' = (1,n')$ with $n'<n$ and $\overline{\zeta} = (1,n)$ it follows that for all $\gamma \in [0,1]$

\[
\phi_1(\gamma) > \sum_{t=2}^b\frac{1}{b-1}\phi_t(\gamma, \overline{\zeta}) \geq \sum_{t=2}^b \psi_t(\gamma)\phi_t(\gamma,\overline{\zeta}').
\]
\end{lemma}
\begin{proof} The first inequality is obvious. For the second one observe that by Lemma~\ref{lem:psiPhi}, $\phi_t(\gamma,\overline{\zeta}') \leq \phi_t(\gamma,\overline{\zeta})$ for all $t$. In turn, showing that 
\[
\sum_{t=2}^b\frac{1}{b-1}\phi_t(\gamma, \overline{\zeta}') \geq \sum_{t=2}^b \psi_t(\gamma)\phi_t(\gamma,\overline{\zeta}')
\]
suffices. Fix a $\gamma \in [0,1]$ and observe that since
\[
1= \sum_{t=2}^b\frac{1}{b-1} = \sum_{t=2}^b\psi_t(\gamma)
\]
and $\psi_2(\gamma) < \frac{1}{b-1}$, there must exist $t^*\leq b$ with $\psi_{t}(\gamma) > \frac{1}{b-1}$ for all $t>t^*$. Consequently, 
\[
 \sum_{t=2}^{t^*}\frac{1}{b-1} \geq \sum_{t=2}^{t^*}\psi_t(\gamma) \text{ and }  \sum_{t=t^*+1}^{b}\frac{1}{b-1} \leq \sum_{t=t^*+1}^{b}\psi_t(\gamma).
\]
Since $\phi_t(\gamma)$ is decreasing in $t$ the claim follows.
\end{proof}

\noindent
Next, setting \[
\Delta(\gamma) = \frac{r}{N}\sum_{t=2}^b\psi_t(\gamma)\phi_t(\gamma) - 1
\]
we find that 
\[
\Delta(0) =\frac{r}{N} \sum_{t=2}^n\psi_t(0)\phi_t(0) - 1 = \frac{r}{N} - 1 <0
\]
for all $n>1$.  Thus, a pure contribution strategy exists for all values of $r$ (as we have assumed $r<N$). \\

In terms of mixed strategies, there are plenty of values of $r$ that yield a $\gamma$ with $\Delta(\gamma) = 0$. Since $n>1$ observe that $\Delta(\gamma) > \Delta(1) = \Delta(0) = \frac{r}{N} -1$ for all $\gamma \in (0,1)$. In turn, by Rolle's Theorem there exists at least one local maximum for $\Delta(\gamma)$ in $(0,1)$. Set $\gamma_0$ to denote this maximum and consider $\Delta(\gamma)$ as a function on $r$ and $\gamma$, $\Delta(r,\gamma)$. Observe that since $r$ is a constant in $\Delta(r, \gamma)$ then $\gamma_0$ is a local maxima for all $r$. Since for all $\gamma \in (0,1) $ we have $\Delta(N,\gamma)>\Delta(N,0) = 0$ and $\Delta(0,\gamma^\sharp)<0$ the continuity of $\Delta(r,\gamma)$ on $r$ implies that there exists a unique value $r^\sharp$ with $\Delta(r^\sharp, \gamma_0) = 0$. Moreover, for all $r>r^\sharp$ we get $\Delta(r,\gamma_0) > 0>\Delta(r,1) = \Delta(r,0) $ and, thus, two roots must exist. 
\end{proof}

\end{document}